\documentclass{aamas2013}
\usepackage{graphicx,amsmath}
\usepackage{color}
\usepackage{dsfont}
\newtheorem{theorem}{Theorem}
\newtheorem{lemma}[theorem]{Lemma}
\newtheorem{corollary}[theorem]{Corollary}
\newtheorem{observation}[theorem]{Observation}

\begin{document}
\title{Exact Algorithms for Weighted and Unweighted Borda Manipulation Problems}
\numberofauthors{2}
\author{
\alignauthor
Yongjie Yang\titlenote{Supported by the DFG Cluster of Excellence (MMCI) and the China Scholarship Council (CSC).}\\
\affaddr{Universit\"{a}t des Saarlandes}\\
\affaddr{Saarbr\"{u}cken, Germany}\\
\email{yyongjie@mmci.uni-saarland.de}
\alignauthor
Jiong Guo\titlenote{Supported by the DFG Cluster of Excellence (MMCI).}\\
\affaddr{Universit\"{a}t des Saarlandes}\\
\affaddr{Saarbr\"{u}cken, Germany}\\
\email{jguo@mmci.uni-saarland.de}
}
\maketitle

\begin{abstract}
Both weighted and unweighted Borda manipulation problems have been proved $\mathcal{NP}$-hard.
However, there is no exact combinatorial algorithm known for these problems.
In this paper, we initiate the study of exact combinatorial algorithms for both weighted and unweighted Borda manipulation problems.
More precisely, we propose $O^*((m\cdot 2^m)^{t+1})$~-~time and $O^*(t^{2m})$~-~time\footnote{$O^*()$ is the $O()$ notation with suppressed factors polynomial in the size of the input.} combinatorial algorithms for weighted and unweighted Borda manipulation problems, respectively, where $t$ is the number of manipulators and $m$ is the number of candidates. Thus, for $t=2$ we solve one of the open problems posted by Betzler et al. [IJCAI 2011]. As a byproduct of our results, we show that the {{unweighted Borda manipulation}} problem admits an algorithm of running time $O^*(2^{9m^2\log{m}})$, based on an integer linear programming technique. Finally, we study the {{unweighted Borda manipulation}} problem under single-peaked elections and present polynomial-time algorithms for the problem in the case of two manipulators, in contrast to the $\mathcal{NP}$-hardness of this case in general settings.
\end{abstract}
\section*{Categories and Subject Descriptors}
F.2 [\textbf{Theory of Computation}]: Analysis of Algorithms and Problem Complexity;
G.2.1 [\textbf{Combinatorics}]: Combinatorial algorithms;
J.4 [\textbf{Computer Applications}]: Social Choice and Behavioral Sciences
\section*{General Terms}
Algorithms
\section*{Keywords}
voting systems, Borda manipulation, exact combinatorial algorithm, single-peaked election
\begingroup
\let\clearpage\relax
\section{Introduction}
Voting systems have many applications in a variety of areas, including political election, web spam reduction, multiagent planning, etc. The Borda system, proposed by Jean-Charles de Borda in 1781~\cite{borda}, is one of the most significant voting systems. It is the prototype of scoring systems and many other voting systems. 
The Borda system has been used for
selecting presidential election candidates in some of the Pacific island countries such as Nauru and Kiribati. It also has been shown that the Borda system is a powerful technique for pattern recognition.

Certain issues which have been attracting much attention in voting systems are the strategic behaviors, e.g., one or more than one voter to influence the outcome of the elections by doing some tricks.
By the celebrated Gibbard-Satterthwaite Theorem \cite{Gibbard73,satterthwaite1975}, every reasonable voting system with at least three candidates can be attacked by the voters with providing insincere votes. However, from the viewpoint of complexity theory, if it is $\mathcal{NP}$-hard to determine how to influence the election, one may give up his attacking to the election. From this point, computational complexity could be a reasonable way to protect elections from being attacked. The first study in this direction was conducted by Bartholdi et al. in their seminal paper~\cite{BARTHOLDI89}. Since then, researches on computational complexity of strategic behaviors of voting systems have been opened up (See~\cite{DBLP:conf/ijcai/IanovskiYEW11,DBLP:conf/ijcai/ObraztsovaE11,DBLP:conf/atal/ErdelyiPR11,DBLP:conf/aaai/DaviesKNW11,DBLP:conf/ijcai/BetzlerNW11}~for more details).
Recently, many $\mathcal{NP}$-hard strategic behavior problems have been extensively studied from the view point of exact, exponential-time algorithms, for instance, 
manipulations~\cite{DBLP:journals/jacm/ConitzerSL07}, bribery problems \cite{DBLP:journals/algorithmica/DornS12}, control problems \cite{DBLP:journals/jair/FaliszewskiHHR09}, etc. For more recent development in this direction, we refer to the excellent survey by Betzler
et al.~\cite{DBLP:conf/birthday/BetzlerBCN12}.
We focus on deriving exact combinatorial algorithms for weighted and unweighted Borda manipulation problems, both of which have been proved $\mathcal{NP}$-hard~\cite{DBLP:conf/ijcai/BetzlerNW11,DBLP:conf/aaai/DaviesKNW11}.

\subsection{Preliminaries}
A {\textit{multiset}} $S:=\{s_1, s_2,..., s_{\scriptsize{|S|}}\}$ is a generalization of a set where objects of $S$ are allowed to appear more than one time in $S$, that is, $s_i=s_j$ is allowed for $i\neq j$. An {\textit{element}} of $S$ is one copy of some object. We use $s\in_{+} S$ to denote that $s$ is an element of $S$. The {\textit{cardinality}} of $S$ denoted by $|S|$ is the number of elements contained in $S$. For example, the cardinality of the multiset $\{1,1,1,2,3,3\}$ is $6$. For two multisets $A$ and $B$, we use $A\uplus B$ to denote the multiset containing all elements in $A$ and $B$.
For example, for $A:=\{1,1,2,3,3,4\}$ and $B:=\{1,2,3\}$, $A\uplus B:=\{1,1,1,2,2,3,3,3,4\}$.

Normally, a {\textit{voting system}} can be specified by a set $\mathcal{C}$ of {\textit{candidates}}, a multiset $\Pi_{{\mathcal{V}}}:=\{\pi_{v_1},\pi_{v_2},...,\pi_{v_n}\}$ of {\textit{votes}} casted by a corresponding set ${\mathcal{V}}:=\{{v}_1, {v}_2, \dots , {v}_{n}\}$ of {\textit{voters}}~($\pi_{v_i}$ is casted by ${v}_i$), and a {\textit{voting protocol}} which maps the {\textit{election}} $(\mathcal{C},\Pi_{\mathcal{V}},\mathcal{V})$ to a candidate $w\in \mathcal{C}$ which we call the {\textit{winner}}. Each vote $\pi_v\in_{+} \Pi_{\mathcal{V}}$ is defined by a bijection $\pi_v: \mathcal{C}\rightarrow [|\mathcal{C}|]$~(in some other literature, a vote is defined as a linear order over the candidates), where $[n]$ denotes the set $\{1,2,...,n\}$. The \textit{position} of a candidate $c$ in $\pi_v$ is the value of $\pi_v(c)$. We say a voter $v$ \textit{placing} a candidate $c$ in his/her $x$-th position or a voter $v$ {\textit{fixing}} his/her $x$-th position by the candidate $c$ if $\pi_v(c)=x$. The candidate placed in the highest, that is, the $|\mathcal{C}|$-th, position in $\pi_v$ is called the most preferred candidate of ${v}$, the candidate placed in the second-highest, that is, the $(|\mathcal{C}|-1)$-th, position in $\pi_v$ is called the second preferred candidate of ${v}$, and so on.

In the following, we use $m$ to denote the number of candidates.
A {\textit{Borda protocol}} can be defined by a vector $\langle m-1,m-2,...,0\rangle$. Each voter contributes $m-1$ score to his/her most preferred candidate, $m-2$ score to his/her second preferred candidates, and so on. The winner is a candidate who has the highest total score. Here, we break a tie randomly, that is, if there is more than one candidate having the highest score, the winner will be chosen randomly from these candidates.
In a {\textit{weighted Borda system}}, each voter ${v}$ is associated with a non-negative integer weight $f({v})$ and contributes $f({v})\cdot (m-1)$ score to his/her most preferred candidate, $f({v})\cdot (m-2)$ score to his/her second preferred candidate, and so on. Accordingly, a candidate having the highest total score wins the election. The unweighted Borda system is the special case of the weighted Borda system where each voter has the unit weight of $1$.

For a candidate $c$ and a voter ${v}$, we use $SC_{{v}}(c)$ to denote the score of $c$ which is contributed by ${v}$, that is, $SC_{{v}}(c):=f({v})\cdot(\pi_v(c)-1)$ (in an unweighted Borda system, $SC_v(c):=\pi_v(c)-1$).
Let $SC_{{\mathcal{V}}}(c)$ denote the total score of $c$ contributed by voters in $\mathcal{{V}}$, that is, $SC_{{\mathcal{V}}}(c):=\sum_{{v}\in {\mathcal{V}}}{SC_{{v}}(c)}$.

In the settings of {\textit{manipulation}}, we have, in addition to $\mathcal{V}$, a set $\mathcal{V}'$ of voters which are called {\it manipulators}. The manipulators form a coalition and desire to coordinate their votes to make a distinguished candidate win the new election with votes in $\Pi_{\mathcal{V}}\uplus \Pi_{\mathcal{V}'}$, where $\Pi_{\mathcal{V}'}$ is the multiset of votes casted by the manipulators. The formal definitions of the problems studied in this paper are as follows.
\medskip

\noindent{\textbf{Unweighted Borda Manipulation}} (UBM)\\[1mm]
{\it{Input:}} An election $(\mathcal{C}\cup \{p\},\Pi_{\mathcal{V}},\mathcal{V})$ where $p$ is not the winner, and a set $\mathcal{V'}$ of $t$ manipulators.\\[2mm]
{\it{Question:}} Can the manipulators cast their votes $\Pi_{\mathcal{V}'}$ in such a way that $p$ wins the election $(\mathcal{C}\cup \{p\},\Pi_{\mathcal{V}}\uplus \Pi_{\mathcal{V}'}, \mathcal{V}\cup\mathcal{V}')?$\\

\noindent{\textbf{Weighted Borda Manipulation}} (WBM)\\[1mm]
{\it{Input:}} A weighted election $(\mathcal{C}\cup \{p\},\Pi_{\mathcal{V}},{\mathcal{V}},f_1:{\mathcal{V}}\rightarrow \mathbb{N})$ where $p$ is not the winner, a set ${\mathcal{V'}}$ of $t$ manipulators and a weight function $f_2:{\mathcal{V'}}\rightarrow \mathbb{N}$.\\[2mm]
{\it{Question:}} Can the manipulators cast their votes $\Pi_{\mathcal{V}'}$ in such a way that $p$ wins the weighted election $(\mathcal{C}\cup \{p\},\Pi_{\mathcal{V}}\uplus \Pi_{\mathcal{V}'}',{\mathcal{V}}\cup {\mathcal{V}'},f:{\mathcal{V}}\cup {\mathcal{V}'}\rightarrow \mathbb{N})$, where $f({v})=f_1({v})$ if ${v}\in {\mathcal{V}}$ and $f({v})=f_2({v})$ otherwise?
\smallskip

Since we break ties randomly, in order to make $p$ the winner, the manipulators must assure that after the manipulation the distinguished candidate $p$ becomes the only candidate who has the highest total score among all candidates.

\subsection{Related Work and Our Contribution}
As one of the most prominent voting systems, complexity of strategic behaviors under the Borda systems has been intensively studied. It is known that many types of bribery and control behaviors under the unweighted Borda system are $\mathcal{NP}$-hard. 
For manipulation behaviors, WBM is $\mathcal{NP}$-hard even when the election contains only three candidates~\cite{DBLP:journals/jacm/ConitzerSL07}. Bartholdi et al.~\cite{BARTHOLDI89} showed that both UBM and WBM in the case of only one manipulator are polynomial-time solvable. The complexity of UBM in the case of more than one manipulator remained open for many years, until very recently it was proved $\mathcal{NP}$-hard even when restricted to the case of only two manipulators~\cite{DBLP:conf/ijcai/BetzlerNW11,DBLP:conf/aaai/DaviesKNW11}. Heuristic and approximation algorithms for UBM have been studied in the literature~\cite{DBLP:journals/ai/ZuckermanPR09,DBLP:conf/aaai/DaviesKNW11}. It is worthy to mention that Zuckerman et al.~\cite{DBLP:journals/ai/ZuckermanPR09} showed that UBM admits an approximation algorithm
which can output a success manipulation with $t+1$ manipulators whenever the given instance has a success manipulation with $t$ manipulators.
By applying the integer linear programming (ILP) techinique, UBM can be solved exactly with a very high computational complexity $O^*(m!^{O(m!)})$ \cite{DBLP:conf/ijcai/BetzlerNW11}.
However, no purely combinatorial exact algorithm seems known for these problems. In particular, Betzler et al. \cite{DBLP:conf/ijcai/BetzlerNW11} posed as an open problem whether UBM can be solved exactly with a running time single-exponentially depending on $m$ in the case of constant number of manipulators.

We propose two algorithms solving WBM and UBM in $O^*((m\cdot 2^m)^{t+1})$ time and $O^*(t^{2m})$ time, respectively, where $t$ is the number of manipulators and $m$ is the number of candidates. Both algorithms are based on dynamic programming techniques. Our results imply that both WBM and UBM can be solved in time single exponentially on $m$ in the case of constant number of manipulators, answering the open question in~\cite{DBLP:conf/ijcai/BetzlerNW11}. Additionally, we improve the running time of the ILP-based algorithms for UBM to $O^*(2^{9m^2\log{m}})$. The key here is to transfer UBM to an integer linear programming with $m^2$ variables. Furthermore, we study polynomial-time algorithms for UBM in the case of at most two manipulators under single-peaked elections~(the definition of single-peaked election is introduced in Sec.~\ref{singlepeaked}). Due to lack of space, some proofs are deferred to the long version.

\section{Algorithm for Weighted Borda Manipulation}\label{sec:wbm}
In this section, we present an exact combinatorial algorithm for WBM.
The following observation is clearly true.
\begin{observation}\label{ob1}
Every true-instance of WBM has a solution where each manipulator places the distinguished candidate $p$ in his/her highest position.
\end{observation}

Let $((\mathcal{C}\cup\{p\},\Pi_{\mathcal{V}},\mathcal{V},f_1),\mathcal{V}',f_2,t)$ be the given instance. Due to Observation~\ref{ob1}, there must be a solution $\Pi_{\mathcal{V}'}$ with $SC_{\mathcal{V}\cup \mathcal{V}'}(p):=SC_{\mathcal{V}}(p)+\sum_{{v'}\in {\mathcal{V}'}}f({v'})\cdot |\mathcal{C}|$ if this instance is true. Therefore, to make $p$ the winner,
$SC_{\mathcal{V}'}(c)\leq g(c)$ should be satisfied for all $c\in \mathcal{C}$, where $g(c):=SC_{\mathcal{V}}(p)+\sum_{{v'}\in {\mathcal{V}'}}f({v'})\cdot |\mathcal{C}|-SC_{\mathcal{V}}(c)-1$. The value of $g(c)$ is called the {\it{capacity}} of $c$. Meanwhile, if in the given instance there is a candidate $c$ with $g(c)<0$, then the given instance must be a false-instance. 
Therefore, in the following, we assume that the given instance contains no candidate $c$ with $g(c)<0$. Based on these, we can reformulate WBM as follows:
\medskip

\noindent\textbf{{Reformulation of WBM}}\\[1mm]
{\it{Input:}} A set $\mathcal{C}$ of candidates associated with a capacity function $g: \mathcal{C}\rightarrow \mathbb{N}$, and a multiset $F:=\{f_1,f_2,...,f_t\}$ of non-negative integers.\\[2mm]
{\it{Question:}} Is there a multiset ${\Pi}:=\{\pi_1,\pi_2,...,\pi_{t}\}$ of bijections mapping from $\mathcal{C}$ to $[|\mathcal{C}|]$ such that $\sum_{i=1}^{t}{f_i\cdot (\pi_i(c)-1)\leq g(c)}$ holds for all $c\in \mathcal{C}$?
\medskip

Here, the bijection $\pi_i$ corresponds to the vote casted by the $i$-th manipulator and $f_i\in_+ F$ corresponds to the weight of the $i$-th manipulator (suppose that a fixed ordering over the manipulators is given).

Our algorithm is based on a dynamic programming method which is associated with a boolean dynamic table defined as $DT(C,Z_1,Z_2,...,Z_t)$, where $C\subseteq \mathcal{C}$ is a subset of candidates, $Z_i\subseteq [|\mathcal{C}|]$ and $|C|=|Z_i|$ for all $i\in [t]$. Here, each $Z_i$ encodes the positions that are occupied by the candidates of $C$ in the vote casted by the $i$-th manipulator. The entry $DT(C,Z_1,Z_2,...,Z_t)=1$ means that there is a multiset ${\Pi}=\{\pi_1,\pi_2,...,\pi_{t}\}$ of bijections mapping from $\mathcal{C}$ to $[|\mathcal{C}|]$ such that for each $i\in [t],\, \bigcup_{c\in C}\{\pi_i(c)\}=Z_i$ and, for every candidate $c\in C$, $c$ is ``safe'' under $\Pi$. Here, we say a candidate $c$ is safe under $\Pi$, if $\sum_{i=1}^{t}{f_i\cdot (\pi_i(c)-1)\leq g(c)}$. Intuitively, $DT(C,Z_1,Z_2,...,Z_t)=1$ means that we can place all candidates of $C$ in the positions encoded by $Z_i$ for all $i\in [t]$ without exceeding the capacity of any $c\in C$. Clearly, a given instance of WBM is a true-instance if and only if $DT(\mathcal{C},Z_1:=[|\mathcal{C}|], Z_2:=[|\mathcal{C}|],...,Z_t:=[|\mathcal{C}|])=1$.
The algorithm is as follows:
\medskip\newline
\noindent{\textbf{Initialization:}} For all $c\in \mathcal{C}$ and $z_1,z_2,...,z_t\in [|\mathcal{C}|]$,
if $\sum_{i=1}^{t}{f_i\cdot (z_i-1)\leq g(c)}$, then $DT(\{c\},\{z_1\},\{z_2\},...,\{z_t\})=1$; otherwise, $DT(\{c\},\{z_1\},\{z_2\},...,\{z_t\})=0$.
\smallskip\newline
\noindent{\textbf{Updating}}: For each $l$ from 2 to $|\mathcal{C}|$, we update the entries $DT(C,Z_1,Z_2,...,Z_t)$ with {${|C|=|Z_1|=|Z_2|=...=|Z_t|=l}$} as follows:
 if $\exists c\in C$ and $\exists z_i\in Z_i~\text{for all}~i\in [t]$ such that $DT(C\setminus\{c\},Z_1\setminus\{z_1\},Z_2\setminus\{z_2\},...,Z_t\setminus\{z_t\})=1$ and $DT(\{c\},\{z_1\},\{z_2\},...,\{z_t\})=1$, 
then $DT(C,Z_1,Z_2,...,Z_t)=1$, otherwise, $DT(C,Z_1,Z_2,...,Z_t)=0$.

\begin{theorem}\label{theorem:wb}
WBM is solvable in $O^*((|\mathcal{C}|\cdot 2^{|\mathcal{C}|})^{t+1})$ time.
\end{theorem}
\begin{proof}
We consider the above algorithm for WBM. In the initialization, we check whether $\sum_{i=1}^{t}{f_i\cdot (z_i-1)\leq g(c)}$ for each candidate $c\in \mathcal{C}$ and each encoded position $z_i\in [|\mathcal{C}|]$ for each $i\in [t]$. Since there are $|\mathcal{C}|$
many candidates and $|\mathcal{C}|$ many positions to be considered for each $z_i$, the running time of the initialization is bounded by $O^*(|\mathcal{C}|^{t+1})$.
In the recurrence, we compute $DT(C,Z_1,Z_2,...,Z_t)$ for all $C\subseteq \mathcal{C}$ and all ${Z_1\subseteq [|\mathcal{C}|]}, {Z_2\subseteq [|\mathcal{C}|]},...,Z_t\subseteq [|\mathcal{C}|]$ with $|C|=|Z_1|=|Z_2|=...=|Z_t|=l$, where $2\leq l\leq m$. To compute each of them, we consider all possibilities of $c\in C$ and ${z_1\in Z_1},{z_2\in Z_2},...,{z_t\in Z_t}$. For each possibility, we further check whether $DT({C\setminus\{c\}},{Z_1\setminus\{z_1\}}, Z_2\setminus\{z_2\},...,Z_t\setminus\{z_t\})=1$ and $DT(\{c\},\{z_1\},\{z_2\},...,\{z_t\})=1$. 
Since there are at most $|\mathcal{C}|^{t+1}$ such possibilities, and there are at most $2^{(t+1)|\mathcal{C}|}$ entries needed to be computed, we arrive at the total running time of ${O^*((|\mathcal{C}|\cdot 2^{|\mathcal{C}|})^{t+1})}$.

The correctness directly follows from the meaning of the dynamic table we defined.
\end{proof}

In~\cite{DBLP:conf/ijcai/BetzlerNW11}, Betzler et al. posed as an open question whether UBM in case of two manipulators can be solved in less than $O^*(|\mathcal{C}|!)$ time. By Theorem~\ref{theorem:wb}, we can answer this question affirmatively.
\begin{corollary}
WBM~(UBM is a special case of WBM) in case of two manipulators can be solved in $O^*(8^{|\mathcal{C}|})$ time.
\end{corollary} 
\section{Algorithm for Unweighted Borda Manipulation}
In this section, we study the UBM problem. Recall that UBM is a special case of WBM where all voters have the same unit weight. However, compared to the weighted version, when we compute $SC_{\mathcal{V}'}(c)$ for a candidate $c$, it is irrelevant which manipulators placed $c$ in the $j$-th positions. The decisive factor is the number of manipulators placing $c$ in the $j$-th positions. This leads to the following approach where we firstly reduce UBM to a matrix problem and then solve this matrix problem by a dynamic programming technique, resulting in a better running time than in Corollary 3.
Firstly, the matrix problem is defined as follows.
\medskip

\noindent{\textbf{Filling Magic Matrix}} (FMM)\\[1mm]
{\it{Input:}} A multiset $g:=\{g_{_1},g_{_2},...,g_m\}$ of non-negative integers and an integer $t>0$.\\[2mm]
{\it{Question:}} Is there an $m\times m$ matrix $M$ with non-negative integers such that $\forall i\in [m],\;\sum_{j=1}^{m}{(j-1)\cdot M[i][j]}\leq g_i$ and $\sum_{j=1}^{m}{M[i][j]}=t$, and $\forall j\in [m],\, \sum_{i=1}^{m}{M[i][j]}=t?$
\medskip

In the following, we present an algorithm for FMM. The algorithm is based on a dynamic programming method associated with a boolean dynamic table $DT(l,T)$, where $l\in [m]$ and $T:=\{T_j\in \mathbb{N}\mid j\in [m], T_j\leq t\}$ is a multiset. The entry $DT(l,T)=1$ means that there is an $m\times m$ matrix $M$ such that: (1) $\sum_{j=1}^{m}M[i][j]= t$ for all $i\in [l]$; (2) $\sum_{i=1}^{l}M[i][j]=T_j$ for all $j\in [m]$; and (3) $\sum_{j=1}^{m}(j-1)\cdot M[i][j]\leq g_i$ for all $i\in [l]$. It is clear that a given instance of FMM is a true-instance if and only if $DT(m,T_{[m]})=1$, where $T_{[m]}$ is the multiset containing $m$ copies of $t$.
The algorithm for solving FMM is as follows.
\medskip\newline
\noindent{\textbf{Initialization:}} For each possible multisets $T:=\{T_j\in \mathbb{N}\mid j\in [m], T_j\leq t\}$, if $\sum_{j=1}^{m}{T_j}=t$ and $\sum_{j=1}^{m}{(j-1)\cdot T_j}\leq g_{_1}$, then set $DT(1,T)=1$; otherwise, set $DT(1,T)=0$.
\smallskip\newline
\noindent{\textbf{Updating:}} For each $l$ from 2 to $t$ and each possible multiset $T:=\{T_j\in \mathbb{N}\mid j\in [m], T_j\leq t\}$, if there is a multiset $T':=\{T'_j\in \mathbb{N}\mid j\in [m], T'_j\leq T_j\}$ such that $DT(l-1,T')=1$, $\sum_{j=1}^{m}{(T_j-T'_j)}=t$ and $\sum_{j=1}^{m}{(j-1)\cdot (T_j-T'_j)}\leq g_{_l}$, then set $DT(l,T)=1$; otherwise, set $DT(l,T)=0$.

\begin{lemma}\label{FMM}
FMM is solvable in $O^*(t^{2m})$ time.
\end{lemma}
\begin{proof}
In the initialization, we consider all possible multisets $T:=\{T_j\in \mathbb{N}\mid j\in [m], T_j\leq t\}$ with $\sum_{j=1}^{m}{T_j}=t$ and $\sum_{j=1}^{m}{(j-1)\cdot T_j}\leq g_{_1}$. Since $T$ has at most $t^m$ possibilities, the running time of the initialization is bounded by $O^*(t^m)$. In the recurrence, we use a loop indicated by a variable $l$ with $2\leq l \leq m$ to update $DT(l, T)$. In each loop we compute the values of the entries $DT(l,T)$ for all multisets $T:=\{T_j\in \mathbb{N}\mid j\in [m], T_j\leq t\}$. To compute each of the entries, we check whether $DT(l-1,T')=1$, $\sum_{j=1}^{m}{(T_j-T'_j)}=t$ and $\sum_{j=1}^{m}{(j-1)\cdot (T_j-T'_j)}\leq g_{_l}$ for all possible multisets $T':=\{T'_j\in \mathbb{N}\mid j\in [m], T'_j\leq T_j\}$. Since there are at most $t^m$ possible multisets $T'$, the time to compute each $DT(l,T)$ is bounded by $O^*(t^m)$. Since $T$ has at most $t^m$ possibilities, there are at most $t^m$ entries needed to be computed in each loop, implying a total time of $O^*(t^{2m})$ for recurrence procedure. In conclusion, the total time of the algorithm is $O^*(t^{2m})$.

The correctness directly follows from the meaning we defined for the dynamic table.
\end{proof}

We now come to show how to solve UBM via FMM.
A {\textit{partial vote}} is a {{partial injection}} $\pi: \mathcal{C}\cup\{p\}\rightarrow [|\mathcal{C}\cup\{p\}|]$ which maps a subset $C\subseteq \mathcal{C}\cup\{p\}$ onto $[|\mathcal{C}\cup\{p\}|]$ such that for any two distinct $a_1, a_2\in C$, $\pi(a_1)\neq \pi(a_2)$. Here, $C$ is the {\textit{domain}} and $\{\pi(a)\mid a\in C\}$ is the {\textit{codomain}} of $\pi$. A position not in the codomain is called a free position. For simplicity, we define $\pi(c)=-1$ for $c\not\in C$. 

\begin{lemma}\label{FMMtoBM}
UBM can be reduced to FMM in polynomial time.
\end{lemma}
\begin{proof}[Sketch] Let $\Gamma:=((\mathcal{C}\cup \{p\},\Pi_{\mathcal{V}},\mathcal{V}), \mathcal{V}', t)$ be an instance of UBM. 
By Observation~\ref{ob1}, we know that if $\Gamma$ is a true-instance there must be a solution $\Pi_{\mathcal{V}'}$ such that each manipulator places $p$ in his/her highest position. We assume that $SC_{\mathcal{V}}(p)+t\cdot |\mathcal{C}|-SC_{\mathcal{V}}(c)-1\geq 0$ for all $c\in \mathcal{C}$ as discussed in Section~\ref{sec:wbm}. Let $(c_1,c_2,...,c_{|\mathcal{C}|})$ be any arbitrary ordering of $\mathcal{C}$. We construct an instance $\Gamma':=(t,g)$ of FMM where $t$ has the same value as in $\Gamma$ and $g:=\{g_1,g_2,...,g_{|\mathcal{C}|}\}$ with $g_i=SC_{\mathcal{V}}(p)+t\cdot |\mathcal{C}|-SC_{\mathcal{V}}(c_i)-1$ for all $i\in [|\mathcal{C}|]$. It is clear that the construction takes polynomial time.
In the following, we prove that $\Gamma$ is a true-instance if and only if $\Gamma'$ is a true-instance.

$\Rightarrow:$ Given a solution $\Pi_{\mathcal{V}'}$ for $\Gamma$, we can get a solution for $\Gamma '$ by setting $M[i][j]=|\{\pi\in_{\tiny{+}} \Pi_{\mathcal{V}'}\mid \pi(c_i)=j\}|$, where $\{\pi\in_{\tiny{+}} \Pi_{\mathcal{V}'}\mid \pi(c_i)=j\}$ is the multiset containing all votes $\pi\in_+ \Pi_{\mathcal{V}'}$ with $\pi(c_i)=j$. By the above construction, the correctness of $M$ is easy to verify.

$\Leftarrow:$ Let a $|\mathcal{C}|\times |\mathcal{C}|$ matrix $M$ be a solution for $\Gamma':=(t,g:=\{g_1,g_2,...,g_{|\mathcal{C}|}\})$. Then, a solution for $\Gamma$, where there are exactly $M[i][j]$ manipulators who place $c_i$ in the $j$-th positions, can be constructed by the following polynomial algorithm.
For simplicity, for a candidate $c_i$ and an integer $j$ with $1\leq j\leq |\mathcal{C}|$, $c_i\rightsquigarrow j$ means that there are less than $M[i][j]$ manipulators who have already placed $c_i$ in the $j$-th position. For two partial votes $\pi$ and $\pi'$ and two candidates $c$ and $c'$, $(\pi,c)\leftrightarrow (\pi',c')$ means to switch the position of $c$ in $\pi$ and the position of $c'$ in $\pi'$, that is, if $\pi(c)=j, \pi'(c')=j'$ then, after $(\pi,c)\leftrightarrow (\pi',c')$, $\pi(c')=j, \pi'(c)=j'$.
\begin{description}\itemsep=-2pt
\item[Algorithm for constructing a solution for ${\Gamma}$ from $\Gamma'$]
\item[Step 1] Initialize $\Pi_{\mathcal{V}'}:=\{\pi_1,\pi_2,...,\pi_t\}$ of partial votes such that each partial vote has empty domain;
\item[Step 2] Set $\pi_z(p)=|\mathcal{C}|+1$ for all $z\in [t]$;
\item[Step 3] For $\bar{j}=|\mathcal{C}|$ to $1$, do
\item[Step 3.1] While $\exists \pi_z$ where the $\bar{j}$-th position is free, do
\item[Step 3.1.1]  Let $c_i$ be any candidate with $c_i\rightsquigarrow \bar{j}$;
\item[Step 3.1.2]  If $\pi_z(c_i)=-1$, then set $\pi_z(c_i)=\bar{j}$; 
\item[Step 3.1.3]  Else, let $j'=\pi_z(c_i)$ and let $\pi_{z'}$ be a vote with $\pi_{z'}(c_i)=-1$;
\item[Step 3.1.3.1] If the $\bar{j}$-th position of $\pi_{z'}$ is free, then set $\pi_{z'}(c_i)=\bar{j}$
\item[Step 3.1.3.2] Else
\item[Step 3.1.3.2.1] While $\exists j''>\bar{j}$ with $\pi_{z}^{-1}(j'')=\pi_{z'}^{-1}(j')$, do
\item[\textcolor{white}{Step 3.1.3.2.1.1}]  $(\pi_z, \pi_z^{-1}(j'))\leftrightarrow (\pi_{z'},\pi_{z'}^{-1}(j'))$;
\item[\textcolor{white}{Step 3.1.3.2.1.2}]  Let $j'=j''$;
\item[\textcolor{white}{Step}] \textbf{End While 3.1.3.2.1}
\item[Step 3.1.3.2.2] $(\pi_z,\pi_z^{-1}(j'))\leftrightarrow (\pi_{z'},\pi_{z'}^{-1}(j'))$;
\item[Step 3.1.3.2.3] Set $\pi_z(c_i)=\bar{j}$;
\item[\textcolor{white}{Step}] \textbf{End Else 3.1.3.2}
\item[\textcolor{white}{Step}] \textbf{End Else 3.1.3}
\item[\textcolor{white}{Step}] \textbf{End While 3.1}
\item[\textcolor{white}{Step}] \textbf{End For 3}
\item[Step 4] Return $\Pi_{\mathcal{V}'}$.
\end{description}

Since $\sum_{i=1}^{|\mathcal{C}|}{M[i][\bar{j}]}=t$ and there are exactly $t$ manipulators, there must be a candidate $c_i$ with $c_i\rightsquigarrow \bar{j}$ whenever there is a vote whose $\bar{j}$-th position is free, which guarantees the soundness of Step 3.1.1. Similarly, there must be a $\pi_{z'}$ with $\pi_{z'}(c_i)=-1$ in Step 3.1.3, since, otherwise, $\sum_{j=1}^{|\mathcal{C}|}{M[i][j]}>t$, contradicting that the given instance of FMM is a true-instance. After the switches in the while loop in Step 3.1.3.2.1 and Step 3.1.3.2.2, both $\pi_z$ and $\pi_{z'}$ must fulfill the injection properties of partial votes. See Fig. \ref{fig:fmmubm} for an example of the switches of the while loop in Step 3.1.3.2.1.
\begin{figure}
\includegraphics[width=0.5\textwidth]{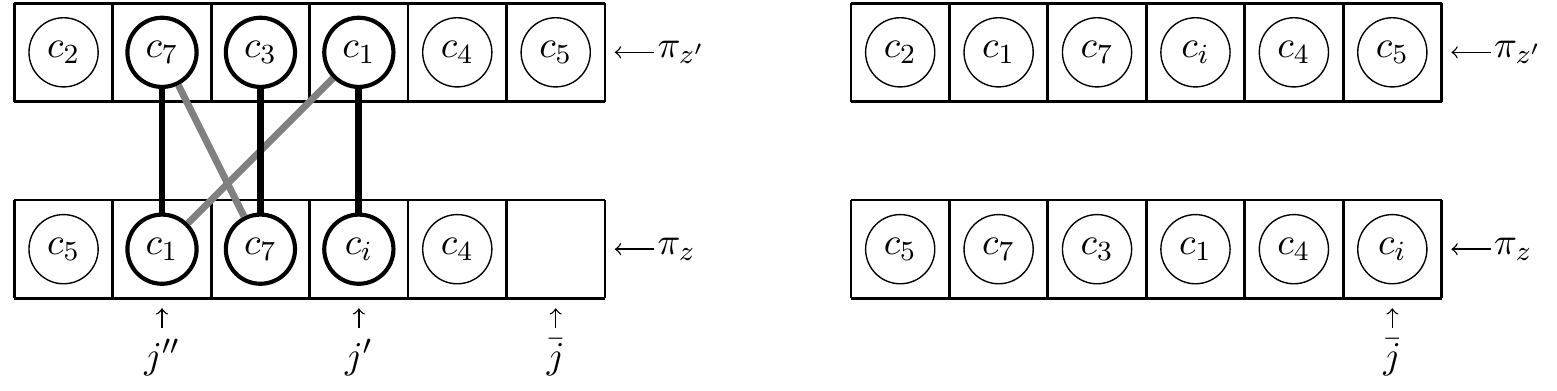}
\caption{The left-hand shows the statues of $\pi_z$ and $\pi_{z'}$ before performing the switches. Due to the algorithm, switchings will happen between the candidates linked by dark lines. The gray lines here are to show that $\pi_{z}^{-1}(j'')=\pi_{z'}^{-1}(j')$, as in the precondition of the while loop in Step 3.1.3.2.1.
The right-hand shows the status after these switches.}\label{fig:fmmubm}
\end{figure}

 Obviously, such a constructed $\Pi_{\mathcal{V}'}$ corresponds to a solution for UBM: for each candidate $c_i\in \mathcal{C}$, $SC_{\mathcal{V}'}(c_i)=\sum_{j=1}^{|\mathcal{C}|}{(j-1)\cdot M[i][j]}\leq g_i= SC_{\mathcal{V}}(p)+t\cdot |\mathcal{C}|-SC_{\mathcal{V}}(c_i)-1$.

To analysis the running time of the algorithm, we need the following lemma.

\begin{lemma}
\label{whileend}
Each execution of the ``while'' loop in Step 3.1.3.2.1 takes polynomial time.
\end{lemma}
\begin{proof}
To prove the lemma, we construct an auxiliary bipartite graph $B$ with $\mathcal{C}_{z'}$ as the left-hand vertices and $\mathcal{C}_z$ as the right-hand vertices, where $\mathcal{C}_{z'}$ and $\mathcal{C}_{z}$ are the sets of candidates which have been placed by $\pi_{z'}$ and $\pi_z$ in some $\bar{j}_{>}$-th position (a $j_>$-th position is a position higher than the $j$-th position), respectively. Two vertices are adjacent if and only if they represent the same candidate (as the vertices linked by a gray line in Figure 1) or they were placed in the same (but not identity) positions~(as the vertices linked by a black line in Figure 1). We observe that the constructed auxiliary graph has maximum degree two. Since $\mathcal{C}_{z'}\setminus \mathcal{C}_z$ is not empty, there is a simple path $P:=(c_{a_1}, c_{a_2},..., c_{a_x})$ with $c_{a_1}=c_i$ and $c_{a_{x}}\in \mathcal{C}_{z'}\setminus \mathcal{C}_z$. It is clear that each execution of the ``while'' loop corresponds to the following switching processing: switching the positions of $a_k$ and $a_{k+1}$ for all $k=1,3,...,x-1$~(since $c_{a_1}=c_i\in\mathcal{C}_z$ and $c_{a_x}\in\mathcal{C}_{z'}$, $x$ is even). The lemma follows from the truth that the length of the simple path is bounded by $2|\mathcal{C}|$.
\end{proof}

We now analysis the running time. The algorithm has three loops. The ``for'' loop in Step 3 has $|\mathcal{C}|$ rounds. The ``while'' loop in Step 3.1 has at most $t$ rounds since each execution of the loop fixes a free position for some $\pi_z$. Due to Lemma \ref{whileend}, the ``while'' loop in Step 3.1.3.2.1 takes polynomial time. Summery all above, the running time of the algorithm is polynomially in $t$ and $|\mathcal{C}|$, which complete the proof.
\end{proof}

Due to Lemmas~\ref{FMM} and~\ref{FMMtoBM}, we get the following theorem.
\begin{theorem}
UBM can be solved in $O^*(t^{2|\mathcal{C}|})$ time.
\end{theorem}

Next we show that FMM can be solved by an integer linear programming (ILP) based algorithm.
The ILP contains $m^2$ variables $x_{ij}$ for $i,j\in [m]$ and, subject to the following restrictions

\[
\left\{\begin{array}{l}
\sum\limits_{i=1}^{m}{x_{ij}=t}~\text{for all}~j\in [m]\\[4mm]
\sum\limits_{j=1}^{m}{x_{ij}=t}~\text{for all}~i\in [m]\\[4mm]
\sum\limits_{j=1}^{m}{(j-1)\cdot x_{ij}}\leq g_i~\text{for all}~i\in [m]\\[4mm]
x_{ij}\geq 0~\text{for all}~i,j\in [m]\end{array}\right.
\]
where $t\in \mathbb{N}$ and $g:=\{g_{_1},g_{_2},\dots,g_m\}$ with $g_i\in \mathbb{N}$ for all $i\in [m]$ are input.

H. W. Lenstra~\cite{lenstra83} proposed an $O^*(\zeta^{O(\zeta)})$-time algorithm for solving ILP with ${\zeta}$ variables. The running time was then improved by R. Kannan~\cite{kannan89}.
\begin{lemma}\label{lemma:ILPFMM}
\cite{kannan89} An ILP problem with $\zeta$ variables can be solved in $O^*(\zeta^{4.5\zeta})$ time.
\end{lemma}

Due to Lemmas~\ref{FMMtoBM} and~\ref{lemma:ILPFMM}, we have the following theorem.

\begin{theorem}
UBM admits an algorithm with running time $O^*(2^{9|\mathcal{C}|^2\log{|\mathcal{C}|}})$.
\end{theorem}

\section{Borda Manipulation under Single-Peaked Elections}\label{singlepeaked}
The single-peaked election was introduced by D. Black in 1948~\cite{Black48}.
The complexity of many voting problems under single-peaked elections has been studied in the literature~\cite{DBLP:conf/tark/FaliszewskiHH11,DBLP:conf/aaai/BrandtBHH10,DBLP:journals/iandc/FaliszewskiHHR11}. It turns out that many
$\mathcal{NP}$-hard problems become polynomial-time solvable when restricted to single-peaked
elections~\cite{DBLP:conf/aaai/BrandtBHH10,DBLP:journals/iandc/FaliszewskiHHR11}.

Given a set $\mathcal{C}$ of candidates and a bijection $\mathcal{L} :\mathcal{C}\rightarrow [|\mathcal{C}|]$, we say that a vote
$\pi_v: \mathcal{C}\rightarrow [|\mathcal{C}|]$ is {\textit{coincident}} with $\mathcal{L}$
if and only if for any three distinct candidates $a,b,c\in \mathcal{C}$ with $\mathcal{L}(a)< \mathcal{L}(b)< \mathcal{L}(c)$
or $\mathcal{L}(c)< \mathcal{L}(b)< \mathcal{L}(a)$,  $\pi_v(c)> \pi_v(b)$ implies $\pi_v(b)> \pi_v(a)$. The candidate having the highest value of $\pi_v(\cdot)$
is called the {\textit{peak}} of $\pi_v$ with respect to $\mathcal{L}$. An election $(\mathcal{C},\Pi_{\mathcal{V}},\mathcal{V})$ is a {\textit{single-peaked}} election if
there exists a bijection $\mathcal{L} :\mathcal{C}\rightarrow [|\mathcal{C}|]$ such that all votes of $\Pi_{\mathcal{V}}$ are coincident with $\mathcal{L}$.
We call such a bijection $\mathcal{L}$ a {\textit{harmonious order}} of $(\mathcal{C},\Pi_{\mathcal{V}},\mathcal{V})$.
See Fig.~\ref{sinlepeaked} for a concrete example of single-peaked elections.

\begin{figure}[t]
\begin{center}
\includegraphics[width=0.3\textwidth]{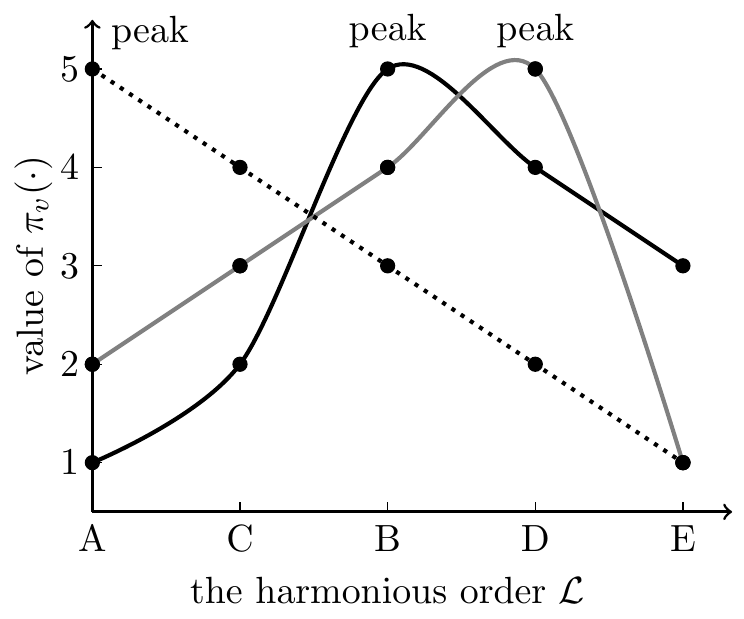}
\end{center}
\caption{This figure shows a visual representation of a single-peaked election which contains five candidates ${\boldsymbol{{A,B,C,D,E}}}$ and three votes $\boldsymbol{\pi_{v_1}}$ (defined by $\boldsymbol{\pi_{v_1}(B)=5}$, $\boldsymbol{\pi_{v_1}(D)=4}$, $\boldsymbol{\pi_{v_1}(E)=3}$,
$\boldsymbol{\pi_{v_1}(C)=2}$, and $\boldsymbol{\pi_{v_1}(A)=1}$), $\boldsymbol{\pi_{v_2}}$ (which is defined by $\boldsymbol{{\pi_{v_2}(D)=5, \pi_{v_2}(B)=4, \pi_{v_2}(C)=3, \pi_{v_2}(A)=2}},$ and $\boldsymbol{{\pi_{v_2}(E)=1}}$), and
$\boldsymbol{\pi_{v_3}}$ (which is defined by $\boldsymbol{\pi_{v_3}(A)=5,}$ $\boldsymbol{\pi_{v_3}(C)=4}$, $\boldsymbol{\pi_{v_3}(B)=3}$,
$\boldsymbol{\pi_{v_3}(D)=2},$  and $\boldsymbol{\pi_{v_3}(E)=1}$). The harmonious order $\boldsymbol{\mathcal{L}}$ is 
$\boldsymbol{\mathcal{L}(A)=1,\mathcal{L}(C)=2,\mathcal{L}(B)=3, \mathcal{L}(D)=4,
\mathcal{L}(E)=5}$.
Here, $\boldsymbol{\pi_{v_1}}$ is illustrated by the dark line, $\boldsymbol{\pi_{v_2}}$ is illustrated by the gray line and $\boldsymbol{\pi_{v_3}}$ is illustrated by the dotted line.
}
\label{sinlepeaked}
\end{figure}

It has been shown in~\cite{DBLP:conf/ecai/EscoffierLO08} that one can test whether a given election is a single-peaked election in polynomial time.
Moreover, a harmonious order can be found in polynomial time
if the given election is a single-peaked election.

For a manipulation problem under single-peaked elections,
we are given a single-peaked election $(\mathcal{C}\cup \{p\},\Pi_{\mathcal{V}},\mathcal{V})$, a harmonious
order $\mathcal{L}$, and a set $\mathcal{V}'$ of manipulators.
We are asked whether the manipulators can cast their votes $\Pi_{\mathcal{V}'}$ in such a way that all their votes are coincident with $\mathcal{L}$ and
$p$ wins the election $(\mathcal{C}\cup \{p\},\Pi_{\mathcal{V}}\uplus\Pi_{\mathcal{V}'},\mathcal{V}\cup \mathcal{V'})$
\cite{DBLP:conf/aaai/Walsh07}. In the following, let UBM1SP and UBM2SP denote the problems of UBM with only one
manipulator and with exactly two manipulators under single peaked elections, respectively.

It is known that the unweighted Borda manipulation problem is polynomial-time solvable with one manipulator \cite{BARTHOLDI89} but becomes $\mathcal{NP}$-hard with two manipulators \cite{DBLP:conf/ijcai/BetzlerNW11,DBLP:conf/aaai/DaviesKNW11}. Here, we show that this problem with two manipulators can be solved in polynomial time in single-peaked elections.
\begin{theorem}\label{SPmain}
Both UBM1SP and UBM2SP are polynomial-time solvable.
\end{theorem}

All remaining parts of this section are devoted to prove Theorem \ref{SPmain}.
To this end, let $({\mathcal{C}\cup \{p\},\Pi_{\mathcal{V}},\mathcal{V}})$ be the given single-peaked election and $\mathcal{L}$ be a harmonious order of $(\mathcal{C}\cup \{p\},\Pi_{\mathcal{V}},\mathcal{V})$.
Let $(l_a,l_{a-1},...,l_1,p,r_1,r_2,...,r_b)$ be an ordering of $\mathcal{C}\cup \{p\}$ with $\mathcal{L}(l_a)<\mathcal{L}(l_{a-1}),...,<\mathcal{L}(l_1)<\mathcal{L}(p)<\mathcal{L}(r_1)<\mathcal{L}(r_2),...,<\mathcal{L}(r_b)$, and let
$\mathcal{C}_L:=\{l_a,l_{a-1},...,l_1\}$ and $\mathcal{C}_R:=\{r_1,r_2,...,r_b\}$.


For two partial votes $\pi_1$ with domain $C_1$ and $\pi_2$ with domain $C_2$, we say $\pi_1$ and $\pi_2$ are {\textit{comparable}} if for every
$c\in C_1\cap C_2$, $\pi_1(c)=\pi_2(c)$. Furthermore, for such comparable partial votes $\pi_1$ and $\pi_2$, let $\pi_1\sqcup \pi_2$ denote the partial vote $\pi$ with domain $C_1\cup C_2$ such that $\pi(c)=\pi_1(c)$ for all $c\in C_1$ and $\pi(c)=\pi_2(c)$ for all $c\in C_2$.

For a partial vote $\pi$ with domain $C$ and a vote $\pi'$, we say $\pi$ is {\textit{extendable}} to $\pi'$ if for any $a\in C$, $\pi(a)=\pi'(a)$.

The definition of single-peaked elections directly implies the following observation.

\begin{observation}\label{obs:singlevote}
Let $\mathcal{L}$ be a harmonious order, $\pi_v$ be a vote which is coincident with $\mathcal{L}$ and $c$ be the peak of $\pi_v$, then
for any $c_1, c_2$ with $\mathcal{L}(c_1)<\mathcal{L}(c_2)<\mathcal{L}(c)$ or $\mathcal{L}(c)<\mathcal{L}(c_2)<\mathcal{L}(c_1)$,
$\pi_v(c_1)<\pi_v(c_2)<\pi_v(c)$.
\end{observation}


\begin{lemma}\label{lem:ob2}
Every true-instance of UBM under single-peaked elections has a solution
where every manipulator places the distinguished candidate $p$ in his/her highest position.
\end{lemma}
\begin{proof}
Assume that $\Pi_{\mathcal{V}'}$ is a solution that contradicts the claim of the lemma,
and ${v}$ is a manipulator who did not place $p$ in his/her highest position.
Without loss of generality, assume that $\mathcal{C}_R\neq \emptyset,\ \mathcal{C}_L\neq \emptyset$
and ${v}$ placed some $l_i$ in his/her highest position.
Let $a:=|\mathcal{C}_L|$ and $b:=|\mathcal{C}_R|$. Due to Observation \ref{obs:singlevote}, $\pi_v(p)>\pi_v(r_1)>\pi_v(r_2)>,...,>\pi_v(r_b)$.
We consider two cases. The first case is that for any $l\in \mathcal{C}_L$, $\pi_v(l)> \pi_v(p)$. In this case, we can create a new solution by recasting $\pi_v$  with
$\pi_v(p)> \pi_v(l_1)> \pi_v(l_2)>,...,> \pi_v(l_a)> \pi_v(r_1)> \pi_v(r_2)>,...,> \pi_v(r_b)$.
The second case is that there is a $l\in \mathcal{C}_L$ with $\pi_v(l)< \pi_v(p)$. In this case, there must be a $z\in [a]$ such that
$\pi_v(l_j)> \pi_v(p)$ for all $j\in [z-1]$ and $\pi_v(l_j)< \pi_v(p)$ for all $a\geq j\geq z$.
Let $C_{small}:=\mathcal{C}_R \cup \{l_z,l_{z+1},...,l_a\}$. We can get a new solution by recasting $\pi_v$ with $\pi'\sqcup \pi''$, where $\pi'$ is the partial vote with domain
$(\mathcal{C}\cup\{p\})\setminus C_{small}$ and codomain $\{|\mathcal{C}\cup\{p\}|, |\mathcal{C}\cup \{p\}|-1,...,|C_{small}|+1\}$ such that
$\pi'(p)> \pi'(l_1)> \pi'(l_2)>,...,> \pi'(l_{z-1})$, and $\pi''$ is the partial vote which has domain $C_{small}$ and is extendable to the original vote $\pi_v$.
Fig. \ref{spcase2}~ illustrates such a case.
\end{proof}

\begin{figure}[t]
\begin{center}
\includegraphics[width=0.45\textwidth]{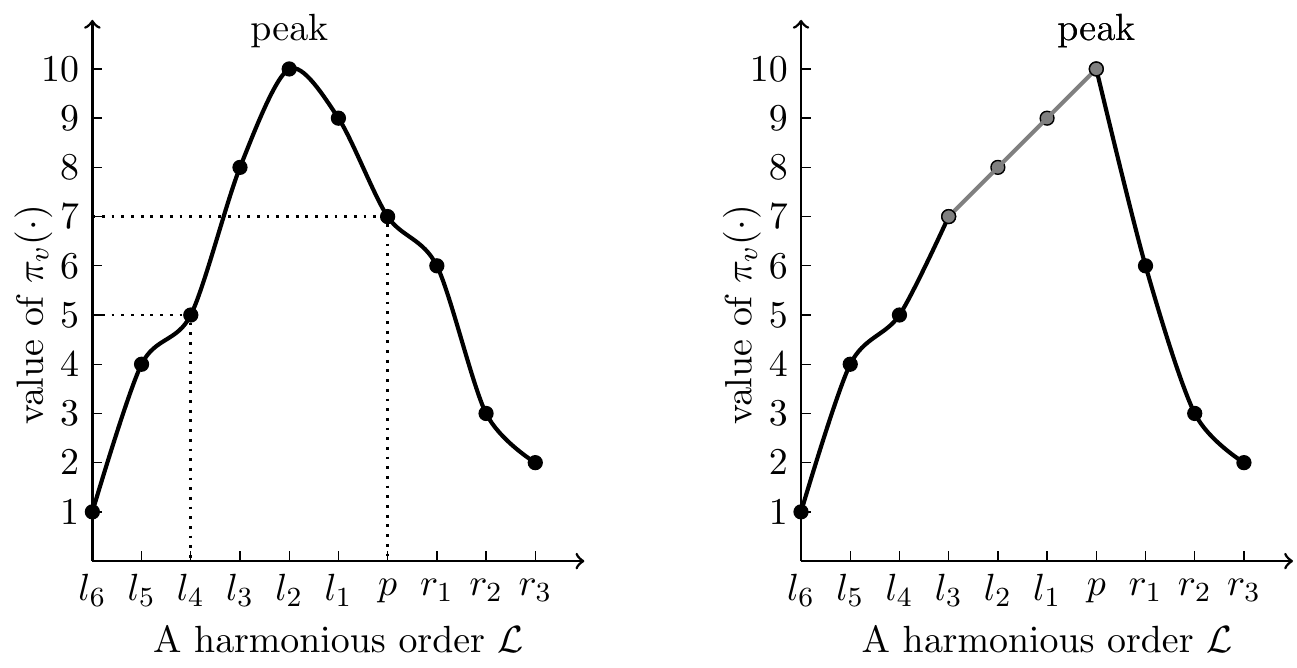}
\end{center}
\caption{Case 2 in the proof of Lemma~\ref{lem:ob2}.
The left-hand shows the original vote
$\boldsymbol{\pi_v}$ with $\boldsymbol{\pi_v(l_2)>\pi_v(l_1)>\pi_v(l_3)>\pi_v(p)>\pi_v(r_1)>\pi_v(l_4)}$ $\boldsymbol{>\pi_v(l_5)>\pi_v(r_2)>\pi_v(r_3)>\pi_v(l_6)}$. The {right-hand} shows the recasted vote $\boldsymbol{\pi'\sqcup \pi''}$ with
$\boldsymbol{\pi_v(p)>\pi_v(l_1)>\pi_v(l_2)>\pi_v(l_3)>\pi_v(r_1)>\pi_v(l_4)}$ $\boldsymbol{>\pi_v(l_5)>\pi_v(r_2)>\pi_v(r_3)>\pi_v(l_6)}$.
Here $\boldsymbol{C_{small}=\{l_4,l_5 ,l_6, r_1,r_2,r_3\}}$, $\boldsymbol{\pi'}$ has domain $\boldsymbol{\{p, l_1, l_2, l_3\}}$ and codomain $\boldsymbol{\{10,9,8,7\}}$, and $\boldsymbol{\pi''}$ has domain
$\boldsymbol{\{r_1,l_4,l_5,r_2,r_3,l_6\}}$ and codomain $\boldsymbol{\{1,2,3,4,5,6\}}$.}
\label{spcase2}
\end{figure}

For a subset $C$ of candidates and a bijection $\mathcal{L}: \mathcal{C}\cup\{p\}\rightarrow [|\mathcal{C}\cup \{p\}|]$, we say that the candidates in $C$ {\textit{lie consecutively}} in $\mathcal{L}$
if and only if there are no $c_1,c_2\in C$ and $c\in(\mathcal{C}\cup \{p\})\setminus C$ such that $\mathcal{L}(c_1)<\mathcal{L}(c)<\mathcal{L}(c_2)$.
If $C$ contains only two candidates $c_1$ and $c_2$, then we say $c_1$ and $c_2$ lie consecutively in $\mathcal{L}$ if $|\mathcal{L}(c_1)-\mathcal{L}(c_2)|=1$.

The following lemma has been proven in~\cite{MichaeTricksinglepeaked89} [Theorem 1].

\begin{lemma}
\label{lem:consecutive}
For two bijections $\mathcal{L}: \mathcal{C}\cup\{p\}\rightarrow [|\mathcal{C}\cup \{p\}|]$ and $\pi: \mathcal{C}\cup\{p\}\rightarrow [|\mathcal{C}\cup \{p\}|]$, $\pi$ is coincident with $\mathcal{L}$ if and only if for any integer $j$ with $1\leq j\leq |\mathcal{C}|$, all candidates in the set $\{c\in\mathcal{C}\cup \{p\}\mid \pi(c)\geq j\}$ lie consecutively in $\mathcal{L}$.
\end{lemma}


The following observation directly follows from the above lemma.

\begin{observation}\label{obs:cons}
Let $\pi$ be a vote with $p$ as the peak. $\mathcal{L}$ is the harmonious order of the given election and $\mathcal{C}_L, \mathcal{C}_R$ are defined as above. Then, $\pi$ is coincident with $\mathcal{L}$ if and only if for all $l_i, l_{i'}\in \mathcal{C}_L$ (resp. $r_j,r_{j'}\in \mathcal{C}_R$), $i<i'$ (resp. $j<j'$) implies $\pi(l_i)>\pi(l_{i'})$ (resp. $\pi({r_j})>\pi(r_{j'})$).
\end{observation}

%

\subsection{Algorithm for UBM1SP}
For two distinct candidates $c,c'\in \mathcal{C}\cup\{p\}$, we say $c$ and $c'$ are neighbors if $|\mathcal{L}(c)-\mathcal{L}(c')|=1$. Let $N(c)$ denote the set of neighbors of $c$. Clearly, every candidate has at most two neighbors. A {\it block} is a subset of candidates lying consecutively in $\mathcal{L}$. For a block $S\subseteq \mathcal{C}\cup\{p\}$, let $N(S):=\{c\in (\mathcal{C}\cup\{p\})\setminus S\mid \exists c'\in S~\text{with}~c\in N(c')\}$. It is easy to verify that $|N(S)|\leq 2$ for every block $S$.

Note that the polynomial-time algorithm for the general Borda manipulation problem with one manipulator proposed in \cite{BARTHOLDI89} cannot be directly used for solving UBM1SP, since UBM1SP requires that the manipulator's vote should be single-peaked. Our polynomial-time algorithm for UBM1SP is a slightly modified version of the algorithm in~\cite{BARTHOLDI89}. Basically, the algorithm places candidates one-by-one in the positions of the manipulator's vote, from the highest to the lowest. The currently highest, unoccupied position of the vote is called the ``next free position''. The details of the algorithm is as follows: (1) Place $p$ in the highest position; (2) Set $S=\{p\}$, where the block $S$ is used to store all candidates which have been placed in some positions by the manipulator; (3) If none of $N(S)$ can be ``safely'' placed in the next free position, then return ``No''. A candidate $c$ can be safely placed in the $j$-th position of the manipulator if $SC_\mathcal{V}(c)+j-1<SC_{\mathcal{V}}(p)+|\mathcal{C}|$. The final score of $p$ is clearly $SC_{\mathcal{V}}(p)+|\mathcal{C}|$. (4) Otherwise, place a candidate $c\in N(S)$, which can be safely placed in the next free position, in the next free position and set $S:=S\cup \{c\}$. If $S=\mathcal{C}\cup \{p\}$, return ``Yes'', otherwise, go back to step (3).

Clearly, the above algorithm needs $O(m)$ time, where $m$ is the number of candidates: Each iteration extends the block $S$ by adding one new candidate to $S$. Since $S$ can be extended at most $|\mathcal{C}|$ times and each iteration takes constant time, the total time is bounded by $O(m)$. The correctness of the above algorithm is easy to verify by Lemma \ref{lem:ob2} and Lemma \ref{lem:consecutive}.
\subsection{Algorithm for UBM2SP}
The algorithm for UBM2SP is similar to the one for UBM1SP: Greedily fill free positions of both manipulators until no free position remains or no candidate can be safely placed in a next free position.
\smallskip

\noindent{\textbf{Main Idea:}} Let $\pi_1$ and $\pi_2$ be the votes of the two manipulators $v_1$ and $v_2$. We use $S_1$ and $S_2$ to denote the sets of candidates that are already placed in some positions of $\pi_1$ and $\pi_2$, respectively. At the beginning, $S_1=S_2=\{p\}$ and we then iteratively extend $S_1$ and $S_2$ by placing candidates in $\pi_1$ and $\pi_2$. After each iteration, we will assure the invariant that $S_1=S_2$ and $S_1$ is a block of $\mathcal{L}$. In each iteration, define $S=S_1=S_2$ and let $l$ and $r$ be the two neighbors of $S$ in $\mathcal{L}$. In order to keep $\pi_1$ and $\pi_2$ single-peaked, there are only two cases to consider: First, both next free positions of $\pi_1$ and $\pi_2$ are occupied by one of $l$ and $r$; Second, $l$ is placed in the next free position of one of $\pi_1$ and $\pi_2$, and $r$ is placed in the next free position of the other. If we have $c\in \{l,r\}$ satisfying $SC_{\mathcal{V}}(c)+2\alpha<SC_{\mathcal{V}}(p)+2|\mathcal{C}|$, where $\alpha$ is the score contributed by one of $\pi_1$ and $\pi_2$ with placing $c$ in its next free position, then we prefer the first case, set $S_1:=S_1\cup \{c\}$ and $S_2:=S_2\cup \{c\}$, and proceed with the next iteration. Otherwise, we check whether the second case is ``safe'', that is, $SC_{\mathcal{V}}(r)+\alpha<SC_{\mathcal{V}}(p)+2|\mathcal{C}|$ and $SC_{\mathcal{V}}(l)+\alpha< SC_{\mathcal{V}}(p)+2|\mathcal{C}|$. If not, then the given instance is a false-instance; otherwise, $S_1:=S_1\cup \{l\}$ and $S_2:=S_2\cup \{r\}$. Next, we have to restore the invariant that $S_1=S_2$. Observe that $S_1$, $S_2$, and $S_1\cap S_2$ all are blocks in $\mathcal{L}$. Moreover, $S_1\setminus S_2$ and $S_2\setminus S_1$ form two blocks and both are neighbors to $S_1\cap S_2$. We apply here another iteration to consider the candidates in $S_1\setminus S_2$ (or $S_2\setminus S_1$) one-by-one. For each of them in $S_1\setminus S_2$ (or $S_2\setminus S_1$), we place it in the highest safe position in $\pi_2$ (or $\pi_1$) and fill the ``gaps'' that are consecutive free positions with candidates not in $S_1\cup S_2$. More details refer to Step 2.4 of the algorithm ALGO-UBM2SP.

For simplicity, we define some new notations. 
For a candidate $c$ and an integer $s$ with $1\leq s\leq |\mathcal{C}|$, we use $c\rightarrow s$ (or $c\not\rightarrow s$) to denote that $c$ can (or cannot) be safely placed in the $(|\mathcal{C}|+1-s)$-th position of some manipulator. We also use $c\rightarrow \{s_1,s_2\}$ (or $c\not\rightarrow \{s_1,s_2\}$) to denote that $c$ can (or cannot) be safely placed in the $(|\mathcal{C}|+1-s_1)$-th position of one manipulator and in the $(|\mathcal{C}|+1-s_2)$-th position of the other, simultaneously.
For a block $S_i$ corresponding to a manipulator $v_i$ with $i=1,2$ and a candidate $c\in N(S_i)$, we use $extend(S_i,c)$ to denote the following operations: (1) place $c$ in the next free position of $\pi_i$, that is, set $\pi_i(c)=|\mathcal{C}|+1-|S_i|$; and (2) extend $S_i$ with $S_i:=S_i\cup \{c\}$.
\smallskip

\noindent{\textbf{The algorithm ALGO-UBM2SP for UBM2SP:}}\vspace{-5pt}
\begin{description}
\itemsep=-2pt
\item[Step 1] Both manipulators place $p$ in their highest positions. Set $S_1=S_2=\{p\}$.
\item[Step 2.] While $S_1=S_2\neq \mathcal{C}\cup \{p\}$, do let $S=S_1$
\item[Step 2.1]  If $\exists c\in N(S)$ with $c\rightarrow \{|S|,|S|\}$, then $extend(S_1,c)$ and $extend(S_2,c)$;
\item[Step 2.2]  Else if $|N(S)|=1$ and $N(S)\not\rightarrow \{|S|,|S|\}$, then return ``No'';\hspace{2.5cm} $\backslash*$ {\sf{Comments: In the case $N(S)=\{c\}$, $c$ is the only candidate which can be placed in the next free positions without destroying the single-peakedness. Thus, if $c\not\rightarrow \{|S|, |S|\}$}, return ``No''} $*\backslash$
\item[Step 2.3]  Else, let $\{c,c'\}=N(S)$. If $c\rightarrow |S_1|$ and $c'\rightarrow |S_2|$, then $extend(S_1,c)$ and $extend(S_2,c')$, otherwise, return ``No''.
\item[Step 2.4] While $S_1\neq S_2$, do
\item[Step 2.4.1] Let $c$ be any candidate in $N(S_1\cap S_2)$ that has already been placed in a position by only one manipulator $v\in \{v_1,v_2\}$;
\item[Step 2.4.2] While $c\not\rightarrow |S_z|$, where $z=1$ if $v=v_2$ and $z=2$ if $v=v_1$, do
\item[Step 2.4.2.1] If $N(S_z)\setminus \{c\}=\emptyset$, return ``No'';
\item[Step 2.4.2.2] Else, let $c'=N(S_z)\setminus \{c\}$; if $c'\rightarrow |S_z|$, then $extend(S_z,c')$; otherwise, return ``No'';
\item[\textcolor{white}{Step}] \textbf{End While 2.4.2}
\item[Step 2.4.3] $extend(S_z,c)$.
\item[\textcolor{white}{Step}] \textbf{End While 2.4}
\item[\textcolor{white}{Step}] \textbf{End While 2}
\item[Step 3] Return ``Yes''.
\end{description}

To show the correctness of the algorithm, we need the following lemma.
For a vote $\pi$ and two integers $x\leq x'$ with $x,x'\in [|\mathcal{C}|+1]$, let $\mathcal{C}_{\pi}(x,x'):=\{c\in \mathcal{C}\cup \{p\}\mid x\leq \pi(c)\leq x'\}$. Let \[\mathcal{C}(\mathcal{L},-c)=:
\begin{cases}
\mathcal{C}_L & ~\text{if}~c\in \mathcal{C}_R\\
\mathcal{C}_R & ~\text{if}~c\in \mathcal{C}_L
\end{cases}\]

\begin{lemma}\label{lem:last}
Let $\{\pi_1,\pi_2\}$ be a solution for UBM2SP, and $c\in\mathcal{C}$ be a candidate with $\pi_1(c)=x$ and $\pi_2(c)=y$. If there are two integers $x',y'$ such that (1) $x'>x,y'>y$; (2) $SC_{\mathcal{V}}(c)+x'+y'-2< SC_{\mathcal{V}}(p)+2|\mathcal{C}|$; (3) $\mathcal{C}_{\pi_1}(x+1,x')\subseteq \mathcal{C}(\mathcal{L},-c),\; \mathcal{C}_{\pi_2}(y+1,y')\subseteq \mathcal{C}(\mathcal{L},-c)$, then, the following two votes $\pi_1', \pi_2'$ with
\[
\pi_1'(c')=
\begin{cases}
\pi_1(c') & ~\text{if}~ \pi_1(c')>x' ~\text{or}~ \pi_1(c')<x\\
x' &~\text{if}~ c'=c\\
j & ~\text{if}~ x< \pi_1(c')=j+1\leq x'
\end{cases}
\]
\[
\pi_2'(c')=
\begin{cases}
\pi_2(c') & ~\text{if}~ \pi_2(c')>y' ~\text{or}~ \pi_2(c')<y\\
y' &~\text{if}~ c'=c\\
j & ~\text{if}~ y< \pi_2(c')=j+1\leq y'
\end{cases}
\]
must be another solution.
\end{lemma}
\begin{proof}
We can transfer $\pi_1$ to $\pi_1'$ by continually switching the position of $c$ and the position of the candidate $c'$ with $\pi_1(c')=\pi_1(c)+1$, until $c$ is moved to the $x'$-th position. Since $c'\in \mathcal{C}(\mathcal{L},-c)$, due to Observation \ref{obs:cons}, each switching will not change a vote which is coincident with $\mathcal{L}$ to a vote which is not coincident with $\mathcal{L}$. By the same way, we can transfer $\pi_2$ to $\pi_2'$ such that $\pi_2'$ is coincident with $\mathcal{L}$. Since $SC_{\mathcal{V}}(c)+x'+y'-2< SC_{\mathcal{V}}(c)+2|\mathcal{C}|$, and each switch does not increase the total scores of any other candidates, $\{\pi_1',\pi_2'\}$ must be a solution.
\end{proof}

We now come to show the correctness of the algorithm.
We say two partial votes $\pi_1, \pi_2$ with domain $S_1, S_2$ with $|S_1|\leq |S_2|$, respectively, are {\it extendable} to a solution, if there is a solution $\{\pi_1',\pi_2'\}$ such that $\{\pi_1'^{-1}(j),\pi_2'^{-1}(j)\}=\{\pi_1^{-1}(j),\pi_2^{-1}(j)\}$ for all $|\mathcal{C}|+2-|S_1|\leq j\leq |\mathcal{C}|+1$ and $\pi_2^{-1}(j)\in \{\pi_1'^{-1}(j),\pi_2'^{-1}(j)\}$ for all $|\mathcal{C}|+2-|S_2|\leq j <|\mathcal{C}|+2-|S_1|$.

\begin{lemma}
ALGO-UBM2SP solves UBM2SP correctly in $O(m)$ time, where $m$ is the number of candidates.
\end{lemma}
\begin{proof}[Sketch]
Firstly, we consider the correctness. In each step, the algorithm extends $S_1$ (or $S_2$) by adding a candidate from $N(S_1)$ (or $N(S_2)$) to $S_1$ (or $S_2$), which makes $S_1$ (or $S_2$) always be a block. Due to Lemma \ref{lem:consecutive}, if the algorithm returns ``Yes'', the two votes $\pi_1$ and $\pi_2$ must be coincident with $\mathcal{L}$. 
We need the following observations.

\begin{observation}\label{obs:1} If at some point $S_1=S_2=S$, $\pi_1,\pi_2$ are extendable to a solution, and a candidate $c\in N(S)$ satisfies $c\rightarrow \{|S|,|S|\}$, then the partial votes $\pi_1'$ with $\pi_1'(c')=\pi_1(c')$ for $c'\in S$ and $\pi_1'(c)=|\mathcal{C}|+1-|S|$  and $\pi_2'$ with $\pi_2'(c')=\pi_2(c')$ for $c'\in S$ and $\pi_2'(c)=|\mathcal{C}|+1-|S|$ are extendable to a solution.
\end{observation}

\begin{observation}\label{obs:2} If at some point $S_1=S_2=S$ with $N(S)=\{l,r\}$, $\pi_1, \pi_2$ are extendable to a solution, no candidate $c\in N(S)$ satisfies $c\rightarrow \{|S|,|S|\}$,  and $l\rightarrow |S|$ and $r\rightarrow |S|$, then, the partial votes $\pi_1'$ with $\pi_1'(c')=\pi_1(c')$ for $c'\in S$ and $\pi_1'(l)=|\mathcal{C}|+1-|S|$ and $\pi_2'$ with $\pi_2'(c')=\pi_2(c')$ for $c'\in S$ and $\pi_2'(r)=|\mathcal{C}|+1-|S|$ are extendable to a solution.
\end{observation}

\begin{observation}\label{obs:3} If at some point two partial votes $\pi_1$ and $\pi_2$ with domains $S_1$ and $S_2$, respectively, are extendable to a solution and there is a candidate $c$ which has been placed in some position in $\pi_{i}$ (but has not been placed by the other vote, that is, $\pi_{3-i}$), where $i=1$ or $i=2$, then, the partial votes $\pi_i$ and  $\pi'$ with $\pi'(c')=\pi_{_{3-i}}(c')$ for $c'\in S_{3-i}$, and $\pi'(c)=\min\{SC_\mathcal{V}(p)+2|\mathcal{C}|-SC_{\mathcal{V}}(c)-\pi_i(c)+1, |\mathcal{C}|+1-|S_{3-i}|\}$ are extendable to a solution.
\end{observation}

The correctness of the above observations follows from Lemma \ref{lem:last}.

Due to Lemma \ref{lem:ob2}, Step 1 is correct. Then, we consider the while loop in Step 2. Due to Observation \ref{obs:1}, if Step 2.1 was executed and the given instance is a true-instance, then, the partial votes after the execution of Step 2.1 must be extendable to a solution. Otherwise, if Step 2.1 was not executed, and $N(S)$ contains only one candidate, then the given instance must be a false-instance as described in the comments of Step 2.2. Finally, if the $extend$ operations in Step 2.3 was executed, then, due to Observation \ref{obs:2}, these operations are correct. If Step 2.3 returns ``No'', then the given instance must be a false-instance. The reason for this is that, due to Lemma \ref{lem:consecutive} and Observation \ref{obs:2}, if the given instance is a true-instance, then $\pi_1,\pi_2$ must be extendable to a solution and thus, the next free positions must be occupied by $N(S)$. If the conditions of Steps 2.1 to 2.3 are not satisfied, the next free positions can not be safely placed, then we can return ``No''. After the execution of $extend$ operations in Step 2.3, $S_1\neq S_2$, and the algorithm then goes to the while loop in Step 2.4. Since $S_1\cap S_2\neq \emptyset$ (at least $p\in S_1\cap S_2$) and both $S_1$ and $S_2$ are blocks (this is easy to check by the definition of $extend$), there must be at least one candidate $c\in N(S_1\cap S_2)$ which has already been placed in a position by a manipulator $v\in \{v_1,v_2\}$. Due to Observation \ref{obs:3}, if the given instance is a true-instance, then we can place $c$ in the position $y=\min\{SC_\mathcal{V}(p)+2|\mathcal{C}|-SC_{\mathcal{V}}(c)-\pi_i(c)+1, |\mathcal{C}|+1-|S_z|\}$, where $i=1, z=2$ if $v=v_1$ and $i=2, z=1$ if $v=v_2$, to get partial votes which are extendable to a solution. Due to Lemma \ref{lem:consecutive}, all free positions higher than the $y$-th position can only fixed by candidates from $\mathcal{C}(\mathcal{L},-c)\setminus S_z$, and each of these free positions is fixed by an unique candidate from $\mathcal{C}(\mathcal{L},-c)\setminus S_z$. These are exactly what the while loop in 2.4.2 does: fixing the free positions which are higher than the $y$-th position one-by one. After the loop, Step 2.4.3 places $c$ in the $y$-th position. However, if $|\mathcal{C}(\mathcal{L},-c)\setminus S_z|$ is less than the number of free positions which are higher than the $y$-th position, then due to Lemma \ref{lem:consecutive}, the given instance must be a false-instance. Due to the above analysis, each case of Steps 2.1 to 2.4 correctly extends the partial votes to two new partial votes which are also extendable to a solution, or correctly returns ``No''.

To analyze the running time, we need to consider how many times Steps 2.1, 2.3, 2.4.1, 2.4.2.2, 2.4.3 are executed. For simplicity, let $n_\lambda$ denote the number of times that Step $\lambda$ is executed throughout the algorithm. It is clear that $n_{2.4.1}=n_{2.4.3}$. Since each execution of Steps 2.1, 2.3, 2.4.2.2, 2.4.3 extends at least one of $S_1$ and $S_2$ by adding a new candidate to the block or terminates the algorithm, and both $S_1$ and $S_2$ can be extended at most $2|\mathcal{C}|$ times, we conclude that $n_{2.1}+n_{2.3}+n_{2.4.1}+n_{2.4.2.2}+n_{2.4.3}= O(|\mathcal{C}|)$. Since each execution of these Steps needs constant time\footnote{To calculate $N(S_1\cap S_2)$ in Step 2.4.1 in constant time, we need some proper date structure like double-linked queue.}, we arrive at a total running time of $O(|\mathcal{C}|)$.
\end{proof}

\section{Conclusion and Future Work}
In this paper, we initiate the study of exact combinatorial algorithms for Borda manipulation problems. We propose two exact combinatorial algorithms with running times $O^*((m\cdot 2^m)^{t+1})$ and $O^*(t^{2m})$ for WBM and UBM, respectively, where $t$ is the number of manipulators and $m$ is the number of candidates in the given election. These results answer an open problem posed by Betzler et al. \cite{DBLP:conf/ijcai/BetzlerNW11}. In addition, we present an integer linear programming based algorithm with running time $O^*(2^{9m^2\log{m}})$ for UBM. Finally, we study UBM under single-peaked elections and propose polynomial-time algorithms for UBM in case of no more than two manipulators. We mention here that all our algorithms can be used with slight modifications for (positional) scoring systems.

 One future direction could be to improve the presented combinatorial algorithms. Parameterized complexity has been proved a powerful tool to handle ${\mathcal{NP}}$-hard problems. As showed here and in \cite{DBLP:conf/ijcai/BetzlerNW11}, UBM is fixed-parameter tractable with respect to the number $m$ of candidates. A challenging work is to improve the running time to $O^*(2^{m\log{m}})$, or even $O^*(c^m)$ with $c$ being a constant.
Furthermore, the complexity of UBM in the case of more than two manipulators under single-peaked elections remains open.
\endgroup
\date{\today}
\bibliographystyle{plain}
\bibliography{bordamanipulations}
\end{document}